\documentclass{autart}
\usepackage{dsfont}
\usepackage{bbm}
\usepackage{bm}
\usepackage{mathrsfs}
\usepackage{supertabular}
\usepackage{natbib}
\usepackage{color}
\usepackage{graphics} 
\usepackage{latexsym, amsmath,theorem, amsfonts, amssymb,graphicx,array,tabularx}
\usepackage{subfigure}
\usepackage{paralist}
\usepackage{caption}
\usepackage{epstopdf}
\usepackage{tabularx,ragged2e,booktabs,caption}
\usepackage{tikz, epic}
\usepackage{mathtools}
\newtheorem{theorem}{Theorem}
\newtheorem{remark}{Remark}
\newtheorem{lemma}{Lemma}

\newtheorem{example}{Example}
\newtheorem{corollary}{Corollary}
\newtheorem{definition}{Definition}
\newtheorem{proof}{Proof}
\newtheorem{problem}{Problem}
\newtheorem{assumption}{Assumption}

\DeclareFontFamily{OT1}{pzc}{}
\DeclareFontShape{OT1}{pzc}{m}{it}{<-> s * [1.000] pzcmi7t}{}
\DeclareMathAlphabet{\mathpzc}{OT1}{pzc}{m}{it}

\makeatletter

\newcommand{\Rmnum}[1]{\expandafter\@slowromancap\romannumeral #1@}
\makeatother

\newcommand\addtag{\refstepcounter{equation}\tag{\theequation}}

\DeclarePairedDelimiter\floor{\lfloor}{\rfloor}

\begin{document}
\begin{frontmatter}

\title{Attack Allocation on Remote State Estimation in Multi-Systems: Structural Results and Asymptotic Solution}


\author[HKUST]{Xiaoqiang Ren}\ead{xren@connect.ust.hk},
\author[KTH]{Junfeng Wu}\ead{junfengw@kth.se},
\author[UU]{Subhrakanti Dey}\ead{subhrakanti.dey@signal.uu.se},
\author[HKUST]{Ling Shi}\ead{eesling@ust.hk}
%
%
\address[HKUST]{Department of Electronic and Computer Engineering, Hong Kong University of Science and Technology, Hong Kong.}

\address[KTH]{ACCESS Linnaeus Center, School of Electrical Engineering, Royal Institute of Technology, Stockholm, Sweden.}

\address[UU]{Signals and Systems Division, Department of Engineering Sciences, Uppsala University, Uppsala, Sweden}


\begin{abstract}
This paper considers optimal attack attention allocation on remote state estimation in multi-systems.
Suppose there are $\mathtt{M}$ independent systems, each of which has a remote sensor monitoring the system and sending its local estimates to a fusion center over a packet-dropping channel.
An attacker may generate noises to exacerbate the communication channels between sensors and the fusion center.
Due to capacity limitation, at each time  the attacker can exacerbate at most $\mathtt{N}$ of the $\mathtt{M}$ channels.
The goal of the attacker side is to seek an optimal policy maximizing the estimation error at the fusion center. The problem is formulated as a Markov decision process (MDP) problem, and the existence of an optimal deterministic and stationary policy is proved. We further show that the optimal policy has a threshold structure, by which the computational complexity is reduced significantly. Based on the threshold structure, a myopic policy is proposed for homogeneous models and its optimality is established.  To overcome the curse of dimensionality of MDP algorithms for general  heterogeneous models, we further provide an asymptotically (as $\mathtt{M}$ and $\mathtt{N}$ go to infinity) optimal solution, which is easy to compute and implement. Numerical examples are given to illustrate the main results.

\end{abstract}
\end{frontmatter}

\section{Introduction}
\emph{Motivations and backgrounds.} Cyber-physical systems, integrating information technology infrastructures with physical processes, are ubiquitous and usually critical in modern societies. Examples include sensor networks, power grids, water and gas supply systems, transportation systems, water pollution monitoring systems. The use of open communication networks, though enabling more efficient design and flexible implementation, makes cyber-physical systems more vulnerable to attacks~\cite{teixeira2015secure,pasqualetti2015control}. Illustrative examples are Iran's nuclear centrifuges accident~\cite{farwell2011stuxnet} and western Ukraine blackout~\cite{Ukraine16}.

Many research works on attackers' possible  behaviors for cyber-physical systems have been done recently.
Generally speaking, attacks can be classified as either denial of service (DoS) attacks or deception attacks~\cite{amin2009safe}.
DoS attacks, comprising availability of data, are most likely threats~\cite{byres2004myths} due to their easy implementation. DoS attacks in networked control systems are studied in~\cite{amin2009safe}. Optimal off-line DoS attack on remote state estimation over a finite horizon for a single sensor system is investigated in~\cite{zhang2015optimal}.  An interactive decision of sending data by sensor and jamming channel by an attacker for remote state estimation in a zero-sum game setting is studied in~\cite{li2015jamming}, and a similar setting is investigated for a control system in~\cite{gupta2010optimal}. Optimal DoS attacks were also studied in the context of detection~\cite{ren2014optimal}.
Deception attacks, comprising integrity of data, are more subtle. Various types of deception attacks have been studied, for example, replay attacks~\cite{mo2009secure}, stealthy deception attacks~\cite{guo2015linear} and  covert attacks~\cite{teixeira2012attack}.

\emph{Related works and contributions.} In this paper, we consider the DoS attacks. Each sensor monitors a (different) system and sends its estimates to a fusion center over a packet-dropping channel. An attacker is present and is capable of attack a certain number of channels at each time.
When a channel is under attack, the packet arrival rate decreases.
The problem is to study the optimal attack policy to maximize the averaged estimation error at the fusion center.
A threshold structure of optimal policies is proved.
The related works are~\cite{mo2012infinite,ren2014dynamic,leong2015optimality}, which study the structure of sensor scheduling policy. Our work differs from these works as follows. First, our work focuses on multi-systems, while a single sensor scenario is studied in aforementioned three papers. Second, we use a fundamentally different methodology. Specifically, both~\cite{mo2012infinite}~and~\cite{ren2014dynamic} proved the structure results by analyzing the stationary probability distribution of states, which, however, works only in very special and simple cases (e.g., a single sensor case). On the contrary, we resort to the MDP theory, a more general and powerful tool. Although an MDP approach was also adopted in~\cite{leong2015optimality}, the methods used to prove either the existence of optimal stationary and deterministic policy or the threshold structure are significantly different due to the different problem models (multi-systems versus single sensor system, different cost/reward structures\footnote{See the details in Footnote~\ref{footnote:cost}.}). Lastly, we provide an asymptotically optimal policy, which is rather easy to compute and implement.

In summary, the main contributions of this paper are as follows.
\begin{enumerate}
\item The problem of attack on remote state estimation in multi-systems is studied by an MDP formulation. The existence of a deterministic and stationary optimal policy is proved, which means that standard MDP algorithms (e.g., value iteration algorithm) can be utilized to compute the optimal policy. Moreover, a threshold structure of optimal policy is proved, by exploiting which a specialized algorithm may be developed to reduce the computational complexity. By the threshold structure, a myopic policy is proposed and its optimality is established for homogeneous models. The myopic policy is such that the expected reward at the next time is maximized.
\item To overcome the curse of dimensionality of MDP algorithms for general heterogeneous models, we provide an asymptotically optimal index-based policy using the multi-armed bandit theory. Since the indices are computed based on each system \emph{solely}, they are quite easy to compute. The index-based policy is implemented just by comparing these indices. What is more, our numerical examples show that this asymptotically optimal policy works quite well even when the number of total systems is small.
\end{enumerate}

The remainder of this paper is organized as follows. In Section~\ref{section:problem-setup}, the mathematical formulation of the considered problem is given. The main results, including the MDP formulation, existence of a stationary and deterministic optimal policy, threshold structure of the optimal policy and the asymptotically optimal index-based policy, are provided in Section~\ref{section:MainResults}. Numerical examples are given in Section~\ref{section:simulation} to illustrate the main results, after which we conclude the paper in Section~\ref{section:conclusion}. All the proofs are presented in Appendices.

\textit{Notation}: $\mathbb{R}$ ($\mathbb{R}_+$) is
the set of real (nonnegative) numbers and $\mathbb{N}$
the set of nonnegative integer numbers.
$\mathbb{S}_{+}^{n}$ ($\mathbb{S}_{++}^{n}$) is the set of $n$ by $n$ real positive semi-definite (definite) matrices.
For a matrix $X$, we use $\mathrm{Tr}(X), X^{\top}$ and $|X|$ to denote its trace, transpose and spectral radius, respectively.  We write $X\succeq 0$ ($X\succ 0$) if $X\in\mathbb{S}_{+}^{n}$ ($X\in\mathbb{S}_{++}^{n}$). For a vector $x$, denote its $i$-th element as $x_{[i]}$. We use $\circ$ to denote function composition, i.e., for two functions $f$ and $g$, $(f\circ g)(x) = f(g(x))$, and $g^i(x)\triangleq \mathop{\underbrace{g\circ g\circ\cdots \circ g}}\limits_{i \:\mathrm{times}}(x)$
with $g^0(x) \triangleq x$.
Let $\times$ denote Cartesian product. For a set $\mathbb{A}$, define the indicator function as $\mathbf{1}_{\mathbb{A}}(x) = 1$, if $x\in\mathbb{A}$; $0$ otherwise. Let $\bm{Pr}(\cdot) (\bm{Pr}(\cdot|\cdot))$ be the (conditional) probability.
For $x\in\mathbb{R}$, denote by $\floor{x}$ the largest integer less than or equal to $x$. Let $\bm{E}[\cdot]$ be the expectation of a random variable.

\section{Problem Formulation}\label{section:problem-setup}
\subsection{Remote Estimation with Packet-dropping Channels}

\setlength{\unitlength}{1.25mm}
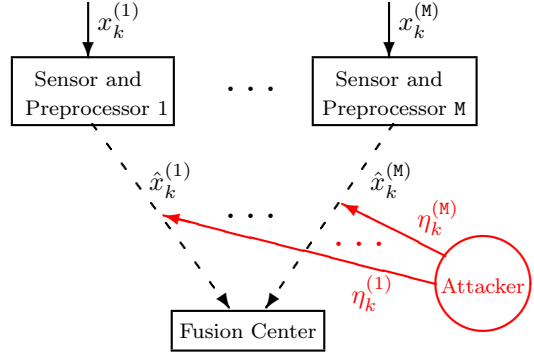
\begin{figure}[tp]
\thicklines
\centering
\begin{picture}
(60,40)(0,-45)
\thicklines

{\Large\put(22.5,-14.5){$\ldots$}}

\put(9,-8){$x_k^{(1)}$}

\put(41,-8){$x_k^{(\mathtt{M})}$}

{\scriptsize
\put(0,-18){\framebox(17,7)}
\put(1.5,-14){ Sensor and }
\put(0.5,-17){ Preprocessor $1$ }
\put(8,-5){\vector(0,-1){6}}

\put(32,-18){\framebox(17.2,7)}
\put(33.5,-14){ Sensor and }
\put(32.1,-17){ Preprocessor $\mathtt{M}$ }
\put(40,-5){\vector(0,-1){6}}

{\color{red}\put(45.5,-36){Attacker}}

\put(17,-42){\framebox(16,4){Fusion Center}}
}

\dashline{0.91}(8.5,-18)(23,-37.5)
\put(23,-37.5){\vector(2,-3){0.1}}

\dashline{0.91}(40.5,-18)(27,-37.5)
\put(27,-37.5){\vector(-2,-3){0.1}}

{\Large\put(22.5,-28){$\ldots$}}

\put(14.5,-25){$\hat{x}_k^{(1)}$}
\put(38,-25){$\hat{x}_k^{(\mathtt{M})}$}

{\color{red}
\put(36,-37){$\eta_k^{(1)}$}
\put(43,-29){$\eta_k^{(\mathtt{M})}$}
\put(50, -35){\circle{10}}
\put(45,-35){\vector(-4,1){29}}
\put(46,-32){\vector(-2,1){11}}
{\Large\put(34,-31){$\ldots$}}
}

\end{picture}
 \caption{ Remote state estimation with  an attacker. } \label{fig:block-diagram}
\end{figure}

There are totally $ \mathtt{M}$ independent discrete-time (i.e., sampled) linear time-invariant systems and $\mathtt{M}$ sensors. The $i$-th sensor monitors the $i$-th system   (Fig.~\ref{fig:block-diagram}):
\begin{subequations}\label{eqn:system-model}
\begin{align}
x^{(i)}_{k+1} & =  A_ix^{(i)}_k + \omega^{(i)}_k,  \\
y^{(i)}_{k} & =  C_ix^{(i)}_{k} + \upsilon^{(i)}_{k},
\end{align}
\end{subequations}
where $x^{(i)}_k \in \mathbb{R}^{n_i}$ is the system state vector and $y^{(i)}_k \in \mathbb{R}^{m_i}$ is the observation vector.
The noises $\omega^{(i)}_{k}$ and
$\upsilon^{(i)}_k $ are i.i.d. white Gaussian random variables
with zero mean and covariance $Q_i\succeq 0, R_i\succ0$, respectively. The initial state $x^{(i)}_0$
is a zero-mean Gaussian random variable
that is uncorrelated with $\omega^{(i)}_k$ and
$\upsilon^{(i)}_k$.  It is assumed that the systems at different sensors are independent of each other. To avoid trivial problems, we assume the systems are unstable, i.e., $|A_i| > 1, \forall i=1,\ldots,\mathtt{M}.$ The pair $(C_i,A_i)$ is assumed to be detectable and $(A_i,Q_i^{1/2})$ stabilizable.

Each sensor is assumed to be intelligent in the sense that a Kalman filter is run locally. With the above detectability and stabilizability assumptions, the estimation error covariance associated with each local Kalman filter converges exponentially to a steady state~\cite{anderson2012optimal}. On the other hand, since
the nature of asymptotic behaviors of remote estimation under malicious attacks (which will be elaborated later)  over an infinite horizon cost  is investigated, without any performance loss, we assume  the Kalman filter at each sensor enters into the steady state at $k=0$. Let the steady state estimation error covariance at sensor $i$ be $\hat{P}^{(i)}$.

At each time  $k$, sensor $i$ sends the output of its local Kalman filter (i.e., the \textit{a posterior} minimum mean square error (MMSE) estimate) $\hat{x}^{(i)}_k$~\cite{anderson2012optimal} 
 to a fusion center over a packet-dropping communication channel. Let $\gamma_k^{(i)}\in\{0,1\}$ denote whether or not the packet is received error-free by the fusion center. If it arrives successfully, $\gamma_k^{(i)}=1$; $\gamma_k^{(i)}=0$ otherwise. Again since the asymptotic behavior over an infinite horizon is studied, it is assumed without any performance loss that $\gamma_0^{(i)}=1,\forall i=1,\ldots,\mathtt{M}$. Since the sensor sends the local MMSE estimates instead of raw measurements, the MMSE estimate and the associated error covariance at the fusion center (whether or not the attacker introduced later is present) for $k\geq 1$ is:
\begin{align*}
\tilde{x}^{(i)}_k=&\left\{
        \begin{array}{ll}
            \hat{x}^{(i)}_k, & \text{if $\gamma_k^{(i)} =1$}, \\
            A_i\tilde{x}^{(i)}_{k-1}, & \text{if $\gamma_k^{(i)} =0$},
        \end{array}
    \right. \\
\tilde{P}^{(i)}_k=&\left\{
        \begin{array}{ll}
            \hat{P}^{(i)}, & \text{if $\gamma_k^{(i)} =1$}, \\
            h_i(\tilde{P}_{k-1}^{(i)}), & \text{if $\gamma_k^{(i)} =0$},
        \end{array}
    \right.
\end{align*}
where functions $h_i, 1\leq i \leq \mathtt{M}$, are defined as follows:
\[h_i(X) = A_iXA_i^{\top} + Q_i, \quad \text{for}\: X\in\mathbb{S}_{+}^{n_i}.\]
Notice that by the assumption $\gamma_0^{(i)}=1,\forall i$, the starting point at the fusion center is: $\tilde{x}^{(i)}_0 = \hat{x}^{(i)}_0$ and $\tilde{P}^{(i)}_0 = \hat{P}^{(i)}$.

\subsection{Attack Model}
There is an attacker capable of  generating noises to exacerbate the communication channels between sensors and the fusion center.
Due to capacity limitation, at each time the attacker can only choose at most $\mathtt{N}$ of the $\mathtt{M}$ channels to attack.
Let $\eta_k^{(i)}\in\{0,1\}$ indicate whether or not the $i$-th channel is under attack: $\eta_k^{(i)}=1$ if it is; $\eta_k^{(i)}=0$ otherwise.
We make the following assumption about the effects of the attacks on packet dropouts.
\begin{assumption}
The packet loss process is memoryless with respect to the considered attacks, i.e., the following equality holds for any $k\geq 1$:
\begin{align*}
\bm{Pr}(\gamma_1^{(i)},\ldots,\gamma_k^{(i)}|\eta^{(i)}_{1:k}    )  = \prod_{j=1}^{k} \bm{Pr}(\gamma_j^{(i)}|\eta_j^{(i)}),
\end{align*}
where $\eta^{(i)}_{1:k} \triangleq (\eta_{1}^{(i)},\ldots,\eta_{k}^{(i)})$. Let $\bm{Pr}(\gamma_k^{(i)} = 1|\eta_k^{(i)}=0) = \epsilon_i$ and $\bm{Pr}(\gamma_k^{(i)} = 1|\eta_k^{(i)}=1) = \underline{\epsilon_i}$. We assume that $0<\underline{\epsilon_i}<\epsilon_i\leq 1$.
\end{assumption}

It is assumed that the attacker has the knowledge of system dynamics (i.e., $A_i, C_i, Q_i$ and  $R_i$\footnote{The steady state estimation error covariance $\hat{P}^{(i)}$ thus can be obtained by solving a discrete-time algebraic Riccati equation.}),
has access to the knowledge of $\{\gamma_k^{(i)}\}_{k\in\mathbb{N}},\forall i=1,\ldots,\mathtt{M},$ and is able to learn the channels' packet arrival rate with or without attacks (i.e., $\underline{\epsilon_i}$ and $\epsilon_i$) from realization of $\{\gamma_k^{(i)}\}_{k\in\mathbb{N}}$.
At each time, the attacker determines the subset of the communication channels to be attacked based on all the information it collects. Let $\gamma_k=(\gamma_k^{(1)}, \ldots,\gamma_k^{(\mathtt{M})})$ and $\gamma_{1:k}=(\gamma_1, \ldots,\gamma_k)$; $\eta_k$ and $\eta_{1:k}$ are defined in the same way. Define a feasible attack attention allocation decision rule at time $k$ as a stochastic kernel $\pi_k$ from $\gamma_{1:k-1}$ and $\eta_{1:k-1}$ to $\Omega$\footnote{
We say $\pi_k$ is a stochastic kernel from $\gamma_{1:k-1}$ and $\eta_{1:k-1}$ to $\Omega$ if the map $\pi_k:\wp(\Omega) \times \{0,1\}^{\mathtt{M}(k-1)}\times \Omega^{k-1} \mapsto [0,1]$ with $\wp(\Omega)$ being the power set of $\Omega$ has the following properties:
\begin{enumerate}
  \item For any realization of $\gamma_{1:k-1}\in\{0,1\}^{\mathtt{M}(k-1)}$ and $\eta_{1:k-1}\in\Omega^{k-1}$, $\pi_k(\cdot|\gamma_{1:k-1},\eta_{1:k-1})$ is a probability measure on $\wp(\Omega)$.
  \item For any set $\mathbb{B}\in\wp(\Omega)$, $\pi_k(\mathbb{B}|\cdot)$ is a measurable function on $\{0,1\}^{\mathtt{M}(k-1)}\times \Omega^{k-1}$.
\end{enumerate}
This kernel-form definition includes the possibility that the attack policy is randomized. Nevertheless, in Section~\ref{section:MainResults} we prove that there exists a deterministic optimal attack policy.},
where $\Omega$ is the set of all feasible $\eta_k$:
\[\Omega \triangleq \left\{\eta\in\{0,1\}^\mathtt{M}: \sum_{i=1}^\mathtt{M} \eta_{[i]} \leq \mathtt{N}   \right\}.  \]

Let $\pi=(\pi_1,\ldots, \pi_k, \ldots)$ be the infinite-horizon attack policy. A policy $\pi$ is feasible only if $\pi_k, k\geq 1$ are feasible. Let $\Pi$ be the set of all feasible policies.  The reward (from the perspective of the attacker) associated with an attack policy $\pi$ is the averaged infinite-horizon estimation error at the centers defined as
\begin{align} \label{eqn:Reward}
\bm{R} (\pi) = \mathop {\lim\inf}_{\mathtt T\to \infty} \frac{1}{\mathtt T} \bm{E}\left[\sum_{k=1}^{\mathtt T} \sum_{i=1}^\mathtt{M} \mathrm{Tr}(\tilde{P}_k^{(i)})\right].
\end{align}
The goal of the attacker is to seek a feasible policy maximizing the above reward:
\begin{problem} \label{problem:1}
\begin{align}
\sup_{\pi\in\Pi}\: \bm{R} (\pi).
\end{align}
\end{problem}
To avoid trivial problems, we assume $\underline{\epsilon_i} > 1- \frac{1}{|A_i|^2}, \forall i$. Otherwise, the attacker may consistently attack the communication channel of the $i$-th system to gain an infinite reward since $\tilde{P}_k^{(i)} \to \infty$ as $k\to\infty$ in the presence of consistent attacks.

\section{Main Results} \label{section:MainResults}
In this section, we solve Problem~\ref{problem:1} by formulating it as a MDP problem. We show that, without any performance loss, the attack decision rule  can be restricted to a smaller class: the optimal policy is deterministic (i.e., the stochastic kernel $\pi_k$ is reduced to a measurable function), stationary (independent of time index $k$) and Markovian (the argument is not the whole history $\gamma_{1:k-1}$). We further prove that the optimal policy has a threshold structure. For the asymptotic regime (i.e., $\mathtt{M}\to\infty$ and $\mathtt{N}\to\infty$), an explicit form of the optimal policy is provided, which is quite easy to compute and implement.

\subsection{MDP Formulation}
Before proceeding, we define a random variable $\tau_k^{(i)}$ as
\[\tau_k^{(i)} = k - \mathrm{max}\{k^*: \gamma_{k^*}^{(i)} =1, 0\leq k^* \leq k\},\]
which indicates the time duration from the last successful transmission time to time $k$. Let $\tau_k=(\tau_k^{(1)},\ldots,\tau_k^{(\mathtt{M})})$.

For ease of exposition, except for the myopic policy and asymptotic analysis, in the remainder of this section we assume that $\mathtt{M}=2$ and $\mathtt{N}=1$. We remark that the following MDP formulation and the existence of a deterministic and stationary optimal policy (Theorem~\ref{theorem:existence}) can be extended trivially to the cases with general $\mathtt{M}$ and $\mathtt{N}$. While for the threshold structure, see Remark~\ref{remark:threshold}.

Now we describe the formulated infinite-horizon discrete-time MDP by a quadruplet ($\mathbb{S}$, $\mathbb{A}$, $\bm{P}(\cdot|\cdot,\cdot)$, $r(\cdot,\cdot)$). Each item in the tuple is elaborated as follows.
\begin{enumerate}
  \item The state at time step $k\geq 1$ is defined as $s_{k} \triangleq (\tau_{k-1}^{(1)},\tau_{k-1}^{(2)})$. Therefore, the state space $\mathbb{S}=\mathbb{N}^2$.
  \item The action space $\mathbb{A} \triangleq \{\mathbf{0},e_1,e_2\}$, where $\mathbf{0}=(0,0)$ means that none of the systems is attacked, $e_1=(1,0)$ and $e_2=(0,1)$ means that \emph{only} the first and \emph{only} the second is attacked, respectively.
  \item The transition probability is stationary. Let $s=(j_1,j_2), s'=(j_1',j_2')$ with $j_i,j_i'\in\mathbb{N}, i=1,2$ and $a\in\mathbb{A}$, then $\forall k\geq 1$,
  \begin{align*}
  \bm{P}(s'|s,a) &\triangleq\bm{Pr}(s_{k+1}=s'|s_k=s,a_k=a)\\
  &\triangleq p_1(j_1'|j_1,a_{[1]})p_2(j_2'|j_2,a_{[2]}),
  \end{align*}
  where for $i=1,2$,
  \begin{align*}
  p_i(j_i'|j_i,a_{[i]})=&\left\{
        \begin{array}{ll}
           \epsilon_i, & \text{if $j_i' =0, a_{[i]} = 0$}, \\
           \underline{\epsilon_i}, & \text{if $j_i' =0, a_{[i]} = 1$}, \\
           1-\epsilon_i, & \text{if $j_i' = j_i + 1, a_{[i]} = 0$}, \\
           1-\underline{\epsilon_i}, & \text{if $j_i' = j_i + 1, a_{[i]} = 1$}, \\
           0, & \text{otherwise}.
        \end{array}
    \right.
\end{align*}
  \item The one-stage reward is independent of the action and defined as
  \begin{align} \label{eqn:OneStageReward}
  r(s=(j_1,j_2),a) = \mathrm{Tr}(h_1^{j_1}(\hat{P}^{(1)})) + \mathrm{Tr}(h_2^{j_2}(\hat{P}^{(2)})).
  \end{align}
\end{enumerate}
Let $\mathbb{H}_k \triangleq (s_1,a_1,\ldots,s_k)$ be the history of states and actions up to time $k$, and $\theta=(\theta_1,\ldots,\theta_k,\ldots)$ be an admissible policy with $\theta_k$ as a
stochastic kernel from $\mathbb{H}_k$ to $\mathbb{A}$. Let $\Theta$ be the class of all such admissible policies.
Define the reward associated with initial state $s_1=s$ and policy $\theta$ by
\[ \bm{J}(s,\theta) =  \mathop{ \lim\inf}_{\mathtt T\to\infty} \frac{1}{\mathtt T}\bm{E}_s^\theta\left[\sum_{k=1}^{\mathtt T} r(s_k,a_k)    \right]. \]
Let $s_{1:k}\triangleq (s_1,\ldots,s_k)$. It is evident that $s_{1:k-1}$ is equivalent to $\gamma_{1:k-1}$, and thus $\theta$ is also equivalent to $\pi$ (specialized to the case $\mathtt{M}=2, \mathtt{N}=1$).
One thus verifies that Problem~\ref{problem:1} (specialized to the case $\mathtt{M}=2, \mathtt{N}=1$) can be equivalently transformed to the following problem.
\begin{problem}
Find the optimal policy $\theta^*\in\Theta$ such that
\[\bm{J}((0,0),\theta^*) = \sup_{\theta\in\Theta} \bm{J}((0,0),\theta).   \]
\end{problem}

\subsection{Structural Results}
We first show that the optimal policy is stationary and deterministic, and satisfies an equality.
We say that $\theta=(\theta_1,\ldots,\theta_k,\ldots)$ is stationary and deterministic,  if there exists a measurable function $f:\mathbb{S}\mapsto \mathbb{A}$ satisfying $\forall k\geq 1$, $\theta_k(f(s)|\mathbb{H}_k')=1$ for any $\mathbb{H}_k' \triangleq (s_1,a_1,\ldots,s_k=s)$. Therefore, in the following, with abuse of notations, we use $f$ to represent a stationary and deterministic policy and
let $\mathbb F$ be the set of all admissible stationary and deterministic policies.
For a measurable function $q:\mathbb S\mapsto \mathbb R$, denote
\begin{equation}\label{eqn:G_function}
\bm{G}(q,s,a)\triangleq \sum_{s'\in\mathbb{S}} q(s')\bm P(s'|s,a).
\end{equation}
We then have the following theorem.
\begin{theorem} \label{theorem:existence}
There exists an optimal stationary and deterministic policy $f^*\in\mathbb{F}$ such that
\begin{align*}
\bm{J} (s,f^*)   \geq \bm{J} (s,\theta), \quad \forall s\in\mathbb{S}, \theta\in\Theta.
\end{align*}
Moreover,
\begin{align*}
f^*(s) =& \mathop{\arg\max}_{a\in\mathbb{A}}\{r(s,a))-\varrho^*+\bm{G}(q,s,a)\}, \addtag \label{eqn:optimalActionExistence} \\
\bm{J} (s,f^*) =& \varrho^*,
\end{align*}
where $q:\mathbb{S}\mapsto \mathbb{R}$ and $\varrho^*\in\mathbb{R}$ satisfy
\begin{align} \label{eqn:differentialEquation}
q(s) = \max_{a\in\mathbb{A}}\{r(s,f(s))-\varrho^*+\bm{G}(q,s,a)\}.
\end{align}
\end{theorem}

Theorem~\ref{theorem:existence} says that deterministic and stationary optimal policy exists and can be computed as~\eqref{eqn:optimalActionExistence} with a differential value function (i.e., $q(s)$) satisfying the Bellman equation~\eqref{eqn:differentialEquation}. This provides a theoretic basis for further analysis (structural properties of optimal policies) and computation methods. In particular, with some additional technical requirements \footnote{One may verify that all requirements in~\cite[Assumption 3.8]{zhu2005value} are satisfied in our case. Due to the limited space, we omit the verification here.}, the value iteration algorithm converges. Furthermore, following the ideas in~\cite[Chapter 8]{sennott2009stochastic}, one can use a value iteration algorithm for finite states to approximate the countable state space in our case, and compute the optimal policy $f^*$, the differential value function $q$ and the optimal averaged reward $\varrho^*$.

We now present a nice structure of the optimal policy $f^*$, which helps reduce the computational complexity of the MDP algorithm significantly.
 \begin{theorem} \label{theorem:OptimalAction}
There exists a critical curve $l_c(j_1,j_2)=0$, of which the function $l_c(j_1,j_2)$ is non-decreasing (and non-increasing) with respect to $j_1$ ($j_2$), dividing $\mathbb{N}^2$ into disjoint regions such that
\begin{enumerate}
  \item $f^*(s=(j_1,j_2)) = e_1$,  if $l_c(j_1,j_2)> 0$;
  \item  $f^*(s=(j_1,j_2)) = e_2$, if $l_c(j_1,j_2) \leq 0$.
\end{enumerate}
\end{theorem}

Due to their ease in implementation and enabling efficient computation, structural results of the optimal deterministic and stationary policy are very much appealing to decision makers~\cite{puterman2005markov}. Thanks to the threshold structure, one only needs to store the transition points \textit{a priori}, and the online implementation is simply by comparisons.  Specialized algorithms can be developed to search among a special class (much smaller) of policies instead of general backward induction algorithms (less efficient)~\cite{puterman2005markov}.


\begin{remark} \label{remark:threshold}
The threshold structure  can be extended to cases with general $\mathtt{M}$ and $\mathtt{N}$. For $1\leq i \leq \mathtt{M}$, define $j_i^- \triangleq (j_1,\ldots,j_{i-1},j_{i+1},\ldots,j_\mathtt{M})$ as the state of the whole system except for the $i$-th system. Then the optimal policy has the following threshold structure. Let state $s=(j_1,\ldots,j_\mathtt{M})$, there exist measurable functions $l_i:\mathbb{N}^{\mathtt{M}-1}\mapsto N$ such that for any $1\leq i \leq \mathtt{M}$, the optimal policy $f^*$ has the form:
\begin{enumerate}
  \item if $j_i \geq l_i(j_i^-)$,  $f^*(s) \in \mathbb E_i$;
  \item if $j_i < l_i(j_i^-)$, $f^*(s) \in \Omega \backslash \mathbb E_i$,
\end{enumerate}
where $\mathbb E_i$ represents the feasible attack attention allocation subset such that the $i$-th system is under attack:
\[\mathbb E_i \triangleq \left\{\eta\in\{0,1\}^\mathtt{M}: \sum_{i=1}^\mathtt{M} \eta_{[i]} \leq \mathtt{N}, \eta_{[i]} =1   \right\}.\]
What is more, the functions $l_i,1\leq i \leq \mathtt{M}$ are such that at each time there are exactly $\mathtt{N}$ systems to be attacked.
\end{remark}

We now consider homogeneous models where the system dynamics are the same and $\epsilon_i, \underline{\epsilon_i}, 1\leq i \leq \mathtt{M}$ are identical.
For the homogeneous models with general $\mathtt{M}$ and $\mathtt{N}$, we propose a myopic policy as follows. \emph{At each time $k$, the attacker attacks the $\mathtt{N}$ systems with largest $\tau_{k-1}^{(i)}$.} Denote this myopic policy by $\pi_{\rm m}$.
Then based on the above threshold structure and the symmetry of homogeneous models, one easily obtains the following corollary, the proof of which is omitted.
\begin{corollary}  \label{Corollary:Myopic}
The myopic policy $\pi_{\rm m}$ is optimal to Problem~\ref{problem:1} for homogeneous models, i.e.,
$\bm{R} (\pi_{\rm m}) = \sup_{\pi\in\Pi}\: \bm {R} (\pi)$.
\end{corollary}
Note that to implement the myopic policy $\pi_{\rm m}$, no specific model knowledge is required.
Instead, one only needs to know the realization of the packet arrival process.

\subsection{Explicit Asymptotic Optimal Policy}
When $\mathtt{M}$ is large, the ``curse of dimensionality" will render MDP numerical algorithms impractical.
Then for heterogeneous models, one may ask whether or not there exists an algorithm that resembles the above myopic policy.  The answer is positive.
In the following, we provide an algorithm that is quite easy to compute and implement. Furthermore, it is proved to be asymptotically optimal as $\mathtt{M}$ and $\mathtt{N}$ go to infinity.

\subsubsection{Virtual Attack Model} \label{section:VirtualAttacker}
To present the algorithm, we introduce an virtual attacker. Consider the $i$-th system \emph{in isolation}. Assume that an (virtual) attacker is able to attack the $i$-th system \emph{all the time}, while
if the attacker refuses to launch an attack at some time, it receives an extra \emph{constant} ``subsidy" $z_i$ (which is independent of the system state $\tau_{k-1}^{(i)}$). In other words, the one-stage reward is given by
\begin{align*}
  r_i(\tau_{k-1}^{(i)},\eta_k^{(i)})
             =   \mathrm{Tr}(h_i^{\tau_{k-1}^{(i)}}(\hat{P}^{(i)})) +  (1 - \eta_k^{(i)})z_i.
\end{align*}
The goal of the attacker is to maximize the averaged infinite-horizon accumulated reward as in Problem~\ref{problem:1} for the sole $i$-th system:
$ \mathop{ \lim\inf}_{\mathtt T\to\infty} \frac{1}{\mathtt T} \bm{E}\left[\sum_{k=1}^{\mathtt T} r_i(\tau_{k-1}^{(i)},\eta_k^{(i)})  \right]$.
Denote the optimal rule for the state $\tau_{k-1}^{(i)} = j$ with $j\in\mathbb{N}$ when the subsidy is $z_i$ as $d_i^*(j,z_i)$\footnote{We use this notation to emphasize the dependence on $z_i$.  It quite easy to show that the optimal rule is stationary, we thus omit the time index $k$.}: $d_i^*(j,z_i) = 0$ if no attacks and $d_i^*(j,z_i) = 1$ otherwise.

To maximize the average infinite-horizon reward for the sole $i$-th system, one can also formulate it as an MDP problem and prove the existence of optimal deterministic and stationary policy. Furthermore, as for Theorem~\ref{theorem:OptimalAction}, one can prove the monotonicity of the differential value function as well, based on which the threshold structure of $d_i^*(j,z_i)$ can be proved. Specifically, for any $1\leq i\leq \mathtt{M}$, given $z_i$, $d_i^*(j,z_i)$ has a form as
\begin{align}  \label{eqn:thresholdAsymptotic}
  d_i^*(j,z_i)
  =&\left\{
        \begin{array}{ll}
           1, & \text{if $j \geq \ell_i(z_i)$}, \\
           0, & \text{if $j < \ell_i(z_i)$},
        \end{array}
    \right.
\end{align}
where $\ell_i(z_i)$ is a function of $z_i$.

\subsubsection{Index-based Policy}
We introduce an index $o_i(\cdot):\mathbb{N} \mapsto \mathbb{R}$ associated with $\tau_{k-1}^{(i)}=j$, which satisfies that, for $1\leq i \leq \mathtt{M}$,
\begin{align*}
&v_i(j)\left[ \frac{1-(1-\epsilon_i)^j}{\epsilon_i} o_i(j)  + \sum_{n=0}^{j} \mathrm{Tr}(h_i^n(\hat{P}^{(i)}))(1-\epsilon_i)^n \right. \\
&\quad \left.+\, (1-\epsilon_i)^j \sum_{n=1}^{\infty} \mathrm{Tr}(h_i^{n+j}(\hat{P}^{(i)}))(1-\underline{\epsilon_i})^n \right]\\
=&v_i(j+1)\left[ \frac{1-(1-\epsilon_i)^{j+1}}{\epsilon_i} o_i(j)  + \sum_{n=0}^{j} \mathrm{Tr}(h_i^n(\hat{P}^{(i)}))(1-\epsilon_i)^n \right. \\
&\quad \left. +\,  (1-\epsilon_i)^{j+1} \sum_{n=0}^{\infty} \mathrm{Tr}(h_i^{n+j+1}(\hat{P}^{(i)}))(1-\underline{\epsilon_i})^n\right], \addtag \label{eqn:AsymptoticIndex}
\end{align*}
where $v_i(j)$ is computed by
\begin{align*}
v_i(j) = \frac{1}{{{\epsilon_i}^{-1}-(1-\epsilon_i)^j}{\epsilon_i}^{-1}  +   (1-\epsilon_i)^j {\underline{\epsilon_i}^{-1} }}.
\end{align*}
Notice that $o_i(\cdot)$ only depends on the $i$-th system and is irrelative with the others. Notice also that $o_i(j)$ in~\eqref{eqn:AsymptoticIndex} can be interpreted as the subsidy such that when the $i$-th system state $\tau_{k-1}^{(i)} = j$, the action ``attack" and ``not attack" are equally attractive if the single $i$-th system is considered.
We propose an index-based policy, denoted by $\pi_{\rm d}$, as follows.  \emph{At each time $k$, the attacker attacks the $\mathtt{N}$ systems of greatest index $o_i(\tau_{k-1}^{(i)})$.} We then have the following theorem.
\begin{theorem} \label{theorem:indexpolicy}
The index-based policy $\pi_{\rm d}$ is asymptotically optimal to Problem~\ref{problem:1}. That is,
as $\mathtt{M}\to\infty$ and $\mathtt{N}\to\infty$ with $\mathtt{N}<\mathtt{M}$,
$\bm {R} (\pi_{\rm d}) \to \bm {R}^*$, where $\bm R^* = \sup_{\pi\in\Pi}\: \bm {R} (\pi)$.
\end{theorem}
\begin{remark}
Numerical simulations in Section~\ref{section:simulation} show that the index-based policy $\pi_{\rm d}$ works quite well even when $\mathtt{M}$ and $\mathtt{N}$ are small.
\end{remark}

\begin{remark}
In some scenarios, the attacker might get a larger reward for attacking one system than the other. Then one may add different weight to attacks on different channels, i.e., the reward in~\eqref{eqn:Reward} is replaced with
\begin{align*}
\bm R (\pi) = \mathop {\lim\inf}_{\mathtt T\to \infty} \frac{1}{\mathtt T} \bm{E}\left[\sum_{k=1}^{\mathtt T} \sum_{i=1}^\mathtt{M} w_i\mathrm{Tr}(\tilde{P}_k^{(i)})\right],
\end{align*}
with $w_i \in\mathbb{R}_+$ being weight coefficients. The main results in this paper, Theorems~\ref{theorem:existence}--\ref{theorem:indexpolicy}, still hold. Amending the reward function by adding into the coefficients, the analysis in the appendices remains valid.

\end{remark}

\section{Numerical Examples} \label{section:simulation}
In this section, we use numerical examples to illustrate the threshold structure of the optimal policy (Theorem~\ref{theorem:OptimalAction}), the optimality of the myopic policy for homogeneous models (Corollary~\ref{Corollary:Myopic}) and the asymptotic optimality of the index-based policy (Theorem~\ref{theorem:indexpolicy}).

\begin{example} \label{example:first}
We let $\mathtt{M}=2$ and $\mathtt{N}=1$. The parameters involved are as follows:
\begin{align*}
 A_1 =&  \left[ \begin{array}{cc}
1.2 & 0.2  \\
0.3 & 1
\end{array} \right],
\quad
 A_2 =  \left[ \begin{array}{cc}
1.2 & 0.15  \\
0 & 1.1
\end{array} \right], \\
Q_1 =&  \left[ \begin{array}{cc}
2 & 0  \\
0 & 1
\end{array} \right],
\quad
Q_2 =  \left[ \begin{array}{cc}
1 & 0.5  \\
0.5 & 0.5
\end{array} \right],
\end{align*}
$C_1 = [1,0], C_2 = [1,0.2], R_1 = 1, R_2 =3, \epsilon_1 = 0.95, \underline{\epsilon_1} = 0.5, \epsilon_2 = 0.9$ and $\underline{\epsilon_2} = 0.4$. Notice that the steady-state local estimation error covariances are
\begin{align*}
\hat{P}^{(1)} =&  \left[ \begin{array}{cc}
0.79 & 0.54  \\
0.54 & 8
\end{array} \right],
\quad
\hat{P}^{(2)} =  \left[ \begin{array}{cc}
1.54 & -0.49  \\
-0.49 & 11.87
\end{array} \right].
\end{align*}
We compute the optimal policy and optimal averaged reward using the value iteration algorithm. To cope with the countable infinity of the state space, the ideas in~\cite[Chapter 8]{sennott2009stochastic} are borrowed. The details of the algorithm are as follows.
We truncate the state space with $N \in\mathbb{N}$, i.e., the truncated state space $\mathbb{S}_N \triangleq \{0,\ldots,N\}^2$. Compute the value function (defined on $\mathbb{S}_N$) iteratively by
\begin{align*}
\bm J^N_n(s) = \max_{a\in\mathbb{A}}\{r(s,a)+
\bm {G}(\bm J^N_{n-1},s,a)\},\: \forall s\in\mathbb{S}_N
\end{align*}
with $\bm J^N_0(s)=0$. Since the value iteration algorithm converges in our case (see~\cite{zhu2005value}), then for any $N$, let
\begin{align*}
 \varrho^*_N \triangleq& \lim_{n\to\infty} \bm J^N_n((0,0))-\bm J^N_{n-1}((0,0)), \\
 q_N(s) \triangleq & \lim_{n\to\infty} \bm J^N_n(s) - \bm J^N_n((0,0)).
\end{align*}
One thus obtain the differential value function $q(s) = \lim_{N\to\infty}q_N(s), \forall s\in\mathbb{S}$. The $N$ is chosen such that $|\varrho^*_N-\varrho^*_{N-1}|/\varrho^*_{N-1}$ is smaller than a prescribed tolerance error. In our simulation, we let $N=19$ and the error is 0.01. We obtain that the optimal averaged reward is $50.21$ and the optimal policy is depicted as in Fig.~\ref{Fig:OptimalPolicy}. One may see that the optimal policy has the threshold structure stated in Theorem~\ref{theorem:OptimalAction}.

\begin{figure}
  \centering
  \includegraphics[scale=0.47]{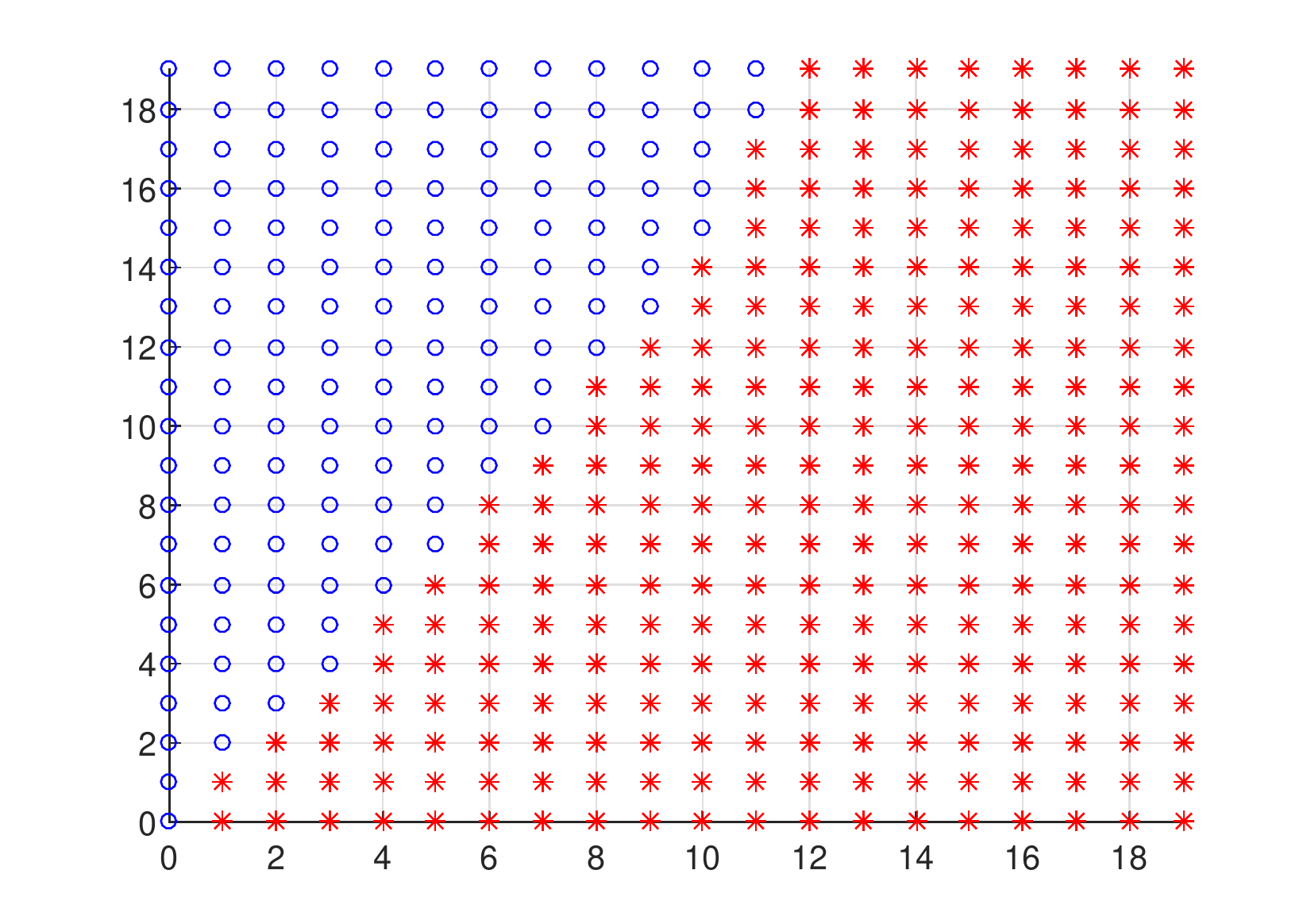}\vspace{-3mm}
  \caption{Optimal action of state $s=(j_1,j_2)$ with x-axis presenting $j_1$ and y-axis $j_2$. The  red stars and blue circles indicate the action $e_1$ and $e_2$, respectively. } \label{Fig:OptimalPolicy}
\end{figure}
\end{example}

\begin{example} \label{example:second}
We shall show that the myopic policy is optimal for homogeneous models. To this end, each system is the same as the $2$nd system in Example~\ref{example:first}, and the state space is also truncated with $N=19$. In the first case, we let $\mathtt{M}=2, \mathtt{N}=1$; the second case $\mathtt{M}=3,\mathtt{N}=2$ and the third case $\mathtt{M}=5,\mathtt{N}=2$. The averaged reward obtained by the MDP algorithm and the myopic policy are shown in Table~\ref{table:homoIndex}. As a baseline, we also simulate a random policy: at each time, $\mathtt{N}$ out of the $\mathtt{M}$ systems are randomly and uniformly chosen to be attacked.   One sees that the averaged rewards obtained by the optimal MDP algorithm and the myopic policy are quite close, which verifies the optimality of the myopic policy. Also, compared with the random policy, the myopic policy has a significant performance improvement.
\begin{table} [!htb]
\center
  \caption{Averaged reward obtained by the MDP algorithm (denote by the symbol $\spadesuit$), the myopic policy ($\clubsuit$) and random policy ($\heartsuit$) in different cases for homogeneous models.}
  \begin{tabular}{ c || c | c| c }
    \hline
    Case No. & $\spadesuit$ & $\clubsuit$ & $\heartsuit$ \\ \hline
    1 & 40.98 & 40.82  & 29.94 \\ \hline
    2 & 71.49 & 71.37 & 55.91 \\ \hline
    3 & 93.75  & 93.46 & 71.45 \\ \hline
  \end{tabular}
 \label{table:homoIndex}
\end{table}
\end{example}

\begin{example}
We do simulations for four cases with heterogeneous models: in the first case, we let $\mathtt{M}=2, \mathtt{N}=1$; the second case $\mathtt{M}=3,\mathtt{N}=2$, the third case $\mathtt{M}=5,\mathtt{N}=2$ and the fourth case $\mathtt{M}=6, \mathtt{N}=3$\footnote{We do not simulate asymptotic cases (i.e., $\mathtt{M}$ and $\mathtt{N}$ are sufficiently large) since state space size increases exponentially with respect to $\mathtt{M}$, the memory required would be beyond our capabilities.}. In each case the first $\floor{\mathtt{M}/2}$ systems are the same as the $1$-st system in Example~\ref{example:first}, while the remaining are the same as the $2$nd system.
We truncate the state space with $N=12$, which is mainly due to computation accuracy of index $o_i(\cdot)$ defined in~\eqref{eqn:AsymptoticIndex}. Specifically, since as $j\to\infty$, $v_i(j)[1-(1-\epsilon_i)^j]/\epsilon_i \to  v_i'(j) [1-(1-\epsilon_i)^{j+1}]/\epsilon_i$, then when $j$ is large enough ($N=13$ for $1$-st system and $N=17$ for $2$-th system), numerical computing software (Matlab in our simulation) cannot provide accurate value of $o_i(\cdot)$.  The averaged reward obtained by the MDP algorithm, the index-based policy and the random policy (the same as in the second example) are shown in Table~\ref{table:heterIndex}, from which one sees that the index-based policy approximates the MDP algorithm surprisingly well even in these non-asymptotic cases. As in Example~\ref{example:second}, the index-based policy has a significant performance gain over the random policy. To better illustrate this performance gain, we further simulate the index-based policy and the random policy for some large $\mathtt{M}$'s and $\mathtt{N}$'s (we do not simulate the MDP algorithm due to capacity limitation). The results are shown in Table~\ref{table:heterIndex2}.

\begin{table} [!htb]
\center
  \caption{Averaged reward obtained by the MDP algorithm (denote by the symbol $\spadesuit$), the index-based policy ($\diamondsuit$) and random policy ($\heartsuit$) in different cases for heterogeneous models.}
  \begin{tabular}{ c || c | c | c }
    \hline
    Case No. & $\spadesuit$ & $\diamondsuit$ & $\heartsuit$ \\ \hline
    1 & 44.88 & 42.72 & 28.15 \\ \hline
    2 & 80.50 & 78.97 & 51.97 \\ \hline
    3 & 106.37  & 103.4 & 69.03  \\ \hline
    4 & 136.22  & 131.94 & 84.5 \\ \hline
  \end{tabular}
 \label{table:heterIndex}
\end{table}

\begin{table}[!htb]
\center
  \caption{Averaged reward obtained by the index-based policy (denote by the symbol $\diamondsuit$) and random policy ($\heartsuit$) in different cases (with large $\mathtt{M}$'s and $\mathtt{N}$'s) for heterogeneous models.}
  \begin{tabular}{ c || c | c }
    \hline
    Case No.  & $\diamondsuit$ & $\heartsuit$  \\ \hline
     1 & 4 446 & 2 816  \\ \hline
     2& 18 971  & 12 658 \\ \hline
    3 & 22 269  & 14 082   \\ \hline
    4 & 25 733  & 16 885   \\ \hline

  \end{tabular}
 \label{table:heterIndex2}
\end{table}
\end{example}

\section{Conclusion} \label{section:conclusion}
In this paper, attack allocation on remote state estimation in multi-systems was considered. The problem was solved by formulating it as an MDP problem, of which an optimal deterministic and stationary policy exists. Threshold structure of the optimal policy was proved, by which both online implementation and off-line computation overhead can be reduced. To overcome the curse of dimensionality, an asymptotically optimal index-based policy, which is quite easy to compute and implement, was provided. The results were verified by numerical simulations. In particular, our numerical examples illustrated that the index-based policy works well even when the number of systems is small. An interesting direction of future works is to investigate the problem in a game-theoretic way, where the sensors (which have limited communication energy) are aware of  the presence of the attacker.


\section*{Appendix A\\ Proof of Theorem~\ref{theorem:existence}}
We first show that our MDP model has some ``nice" properties, by which Theorem~\ref{theorem:existence} can be proved. To this end, we define a function $\bm W: \mathbb{S} \mapsto [1,\infty)$ as
  \begin{align*}
  &\bm  W(s=(j_1,j_2))\\
  &=\left\{
        \begin{array}{ll}
           2, & \text{if $j_1 =0, j_2 = 0$}, \\
           \bm W_1(j_1) + \bm W_2(j_2), & \text{otherwise},
        \end{array}
    \right.  \addtag \label{eqn:FunctionW}
\end{align*}
with $\bm W_1,\bm W_2: \mathbb{N}\mapsto [1,\infty)$ as
  \begin{align*}
  \bm W_1(j)
  &=\left\{
        \begin{array}{ll}
           \phi \lambda_1^j, & \text{if $j \leq N_1$}, \\
           \phi \lambda_1^{N_1} |A_1|^{2(j-N_1)}, & \text{if $j > N_1$},
        \end{array}
    \right.  \\
  \bm W_2(j)
  &=\left\{
        \begin{array}{ll}
           \phi \lambda_2^j, & \text{if $j \leq N_2$}, \\
           \phi \lambda_2^{N_2} |A_2|^{2(j-N_2)}, & \text{if $j > N_2$},
        \end{array}
    \right.
\end{align*}
where $\phi, \lambda_i, N_i$ are parameters satisfying the following:
for each $i=1,2$,
\begin{align*}
\lambda_i >& 1, \\
(1-\underline{\epsilon_i}) ( \lambda_i - 1 )  \leq& \frac{1}{2}\underline{\epsilon_1}\underline{\epsilon_2},  \addtag \label{eqn:lambdai} \\
\phi[\beta - (1-\frac{1}{2}\underline{\epsilon_1}\underline{\epsilon_2})] \geq& 1, \addtag \label{eqn:initialphi} \\
\phi \lambda_i^{N_i} [\beta - (1-\underline{\epsilon_i})|A_i|^2 ] \geq& \phi + 1, \addtag \label{eqn:beta2}
\end{align*}
with a constant $\beta<1$, which is bounded below by
\begin{align} \label{eqn:beta}
\beta > \max\left(1-\frac{1}{2}\underline{\epsilon_1}\underline{\epsilon_2}, (1-\underline{\epsilon_i})|A_i|^2\right),\: i=1,2.
\end{align}
One may see that since $\phi > 1, \lambda_i>1$, $\bm W_i$ together with $\bm W$ are well defined (i.e., they are all greater than $1$).

About $\bm W$, we have the following two lemmas. Before proceeding, we need the following definition.
\begin{definition}
Given a function $\bm W:\mathbb{S} \mapsto [1,\infty)$, for a function $u: \mathbb{S} \mapsto \mathbb{R}$, define its $\bm W-$norm as
\[\|u\|_W = \sup_{s\in\mathbb{S}} |u(s)|/\bm W(s).\]
Let $\mathbb{B}_{\bm W}(\mathbb{S})$ be the normed linear space of measurable functions $u$ on $\mathbb{S}$ with $\|u\|_{\bm W} < \infty.$
\end{definition}

\begin{lemma} \label{lemma:uniformergodic}
For any $f\in\mathbb{F}$, the  transition kernel $\bm P(\cdot|\cdot,f(\cdot))$ is \emph{uniformly} $\bm W-$geometrically ergodic\footnote{Interested readers are referred to~\cite{meyn1993markov} to see a more elegant definition, which, however, requires more background knowledge, and is thus omitted here.},
i.e., for any $f\in\mathbb{F}$ and any measurable function $u\in\mathbb{B}_{\bm W}(\mathbb{S})$, there exists a probability measure $\mu_f$ (depending on $f$) and constants $L$ and $\delta<1$, which are independent of $f$, such that  for any $s\in\mathbb{S}, k\in\mathbb{N}$,
\begin{align} \label{eqn:uniformErgodic}
\left|\bm G(u,s,f(s)) -\int u \mathrm{d}\mu_f\right| \leq \|u\|_{\bm W} \bm W(s) L \delta^k.
\end{align}
\end{lemma}
\begin{proof}
We prove that
for each $f\in\mathbb{F}$,
there exist constant $0<\varpi<1$ and  $b$, which are independent of $f$, such that
\begin{equation}\label{eqn:atom}
\bm P \left((0,0)|(0,0),f((0,0))\right) \geq  \varpi
\end{equation}
and for any $s\in\mathbb S$
\begin{equation}\label{eqn:Lyapnov}
\bm G(\bm W,s, f(s))\leq  \beta \bm W(s) + b \mathbf{1}_{\{(0,0)\}}(s)
\end{equation}
      where $\bm W(\cdot)$ and $\beta$ are defined in~\eqref{eqn:FunctionW}~and~\eqref{eqn:beta}, respectively. Then by~\cite[Theorem 2.1 and 2.2]{meyn1994computable}, for each $f$, $L$ and $\delta$ in ~\eqref{eqn:uniformErgodic} can be chosen in terms of $\varpi, \beta, b$ (which are independent of $f$). The uniform ergodicity in Lemma~\ref{lemma:uniformergodic} thus can be established.

      Equation \eqref{eqn:atom} is trivial.  To show~\eqref{eqn:Lyapnov}, notice that when $s=(0,0)$, one may choose a sufficiently large $b$ such that~\eqref{eqn:Lyapnov} is satisfied. Let $s\triangleq(j_1,j_2)\neq (0,0)$, suppose the action is $e_1$, then
\begin{align*}
&\bm G(\bm W, s,f(s)) \\
= &  (1-\underline{\epsilon_1})\bm W_1(j_1+1) + (1-\epsilon_2)\bm W_2(j_2+1) \\
&+\underline{\epsilon_1}(1-\epsilon_2) \phi + (1-\underline{\epsilon_1})\epsilon_2 \phi + 2\underline{\epsilon_1}\epsilon_2 \\
\leq & (1-\underline{\epsilon_1})\bm W_1(j_1+1) + \underline{\epsilon_1}(1-\epsilon_2) \phi + 1  \addtag \label{eqn:Term1}\\
&+(1-\epsilon_2)\bm W_2(j_2+1) + (1-\underline{\epsilon_1})\epsilon_2 \phi +1 \addtag \label{eqn:Term2}
\end{align*}
Denote the term in~\eqref{eqn:Term1}~and~\eqref{eqn:Term2} by $\Lambda_1$ and $\Lambda_2$, respectively. We show $\Lambda_1 \leq \beta \bm W_1(j_1)$ by examining cases.

\emph{Case $j_1<N_1$: }
\begin{align*}
\Lambda_1 = & (1-\underline{\epsilon_1})\lambda_1 \bm  W_1(j_1) + \underline{\epsilon_1}(1-\epsilon_2) \phi + 1 \\
\leq & (1-\underline{\epsilon_1})\lambda_1 \bm  W_1(j_1) + \underline{\epsilon_1}(1-\epsilon_2) \bm W_1(j_1) + 1 \\
\leq & (1-\frac{1}{2}\underline{\epsilon_1}\underline{\epsilon_2}) \bm W_1(j_1) + 1 \\
\leq & \beta\bm  W_1(j_1),
\end{align*}
where the second inequality follows from~\eqref{eqn:lambdai} and the last one~\eqref{eqn:initialphi}.

\emph{Case $j_1\geq N_1$: }
\begin{align*}
\Lambda_1 = & (1-\underline{\epsilon_1})|A_1|^2  \bm W_1(j_1) + \underline{\epsilon_1}(1-\epsilon_2) \phi + 1 \\
\leq & \beta \bm W_1(j_1),
\end{align*}
where the inequality follows from~\eqref{eqn:beta2}.
Using similar arguments, one may prove $\Lambda_2 \leq \beta \bm W_2(j_2)$, which completes the case when action $e_1$ is used. When $e_2$ or $\mathbf{0}$, similar results can be proved in the same way. The proof thus is complete. \hfill $\square$
\end{proof}

\begin{lemma} \label{lemma:boundedreward}
There exists a constant $\alpha$ such that
\[ \big\|\bar{r}(s)\big\|_{\bm W}\leq \alpha, \]
with $\bar{r}(s)\triangleq\sup_{a\in\mathbb{A}}r(s,a)$.
\end{lemma}
\begin{proof}
Let $\bm W_i'(j) = |A_i|^{2j}, j\in\mathbb{N}, i=1,2$. Since $\bm W(s)\geq 1, \forall s$, we only need to check asymptotic case of $\bar{r}(s)/\bm W(s)$. Since for $i=1,2$,
\[\lim_{j\to\infty} \frac{\bm W_i(j)}{\bm W_i'(j)} = \phi\lambda_i^{N_i} |A_i|^{-2N_i} \]
is a constant,
it suffices to prove for $i=1,2$,
\begin{align} \label{eqn:boundreward1}
\limsup_{j\to\infty} \frac{\mathrm{Tr}(h_i^{j}(\hat{P}^{(i)}))}{\bm W_i'(j)} < \infty.
\end{align}
Since the arguments are exactly the same, we do not distinguish $i=1$ and $i=2$ and suppress subscript $i$ in the remainder of this proof.
Let $\varphi$ be a constant such that $\hat{P} \preceq \varphi I$ and $Q \preceq \varphi I$. Define a function
\[g(X) = AXA^{\top} + \varphi I.\]
One then obtains that
\begin{align*}
h^j(\hat{P}) \preceq& g^j(\varphi I)
\preceq \varphi \sum_{k=0}^j A^k(A^{\top})^k,
\end{align*}
which yields that $\mathrm{Tr}(h^j(\hat{P}))/|A|^{2j}$ is bounded. Equation~\eqref{eqn:boundreward1} thus follows and the proof is complete. \hfill $\square$
\end{proof}

We are ready to prove Theorem~\ref{theorem:existence} using the results in~\cite{guo2006average} \footnote{Notice that the distinguished feature of our MDP model is that the one-stage reward function is \emph{unbounded above}, while the conventional MDP models (including the model in~\cite{leong2015optimality}) have the reward (cost) function being bounded above (below). \label{footnote:cost}}. Since our state space is denumerable, by Remark 4.1(b) thereof, to prove Theorem~\ref{theorem:existence}, it suffices to verify Assumptions 3.1,3.2\footnote{Notice that in~\cite{guo2006average}, the goal is to minimize an average cost, while we aims to maximize a reward function. Assumption 3.2 thereof should be adjusted accordingly, i.e., the requirement that the one-stage cost function is lower semicontinuous should be replaced with that the one-stage reward function is upper semicontinuous.  } and 3.3 thereof. Since our action space is finite, Assumption 3.2 holds trivially. Assumption 3.1 and 3.3 follows directly from Lemma~\ref{lemma:uniformergodic}~and~\ref{lemma:boundedreward} (see Remark 3.3(b) thereof). The proof thus is complete.

\section*{Appendix B\\ Proof of Theorem~\ref{theorem:OptimalAction}}

To present structure of the optimal action, we  give the following supporting lemma about the structure of so-called differential value function $q(s)$ in~\eqref{eqn:differentialEquation}. To this end, we define a partial order on $\mathbb{S}$. Let $s=(j_1,j_2),s'=(j_1',j_2')\in\mathbb{S}$, we say that $s\preccurlyeq s'$ if $j_1\leq j_1'$ and $j_2\leq j_2'$. This partially ordered set is a lattice.
Let $s\uparrow (\downarrow) s'$  denote the join (meet) on $(\preccurlyeq, \mathbb S)$.

\begin{lemma} \label{lemma:differentialEqn}
Let $s,s'\in\mathbb{S}$, for function $q(\cdot)$, the followings hold:
\begin{description}
  \item[Monotonicity:] If $s\preccurlyeq s'$, $q(s)\leq q(s')$.
  \item[Submodularity:] $q(s) + q(s') \geq q(s\downarrow s') + q(s \uparrow s')$.
\end{description}
\end{lemma}

\begin{proof}
Let $0<\alpha<1$. Define the \emph{discounted} reward associated with the initial state $s_1=s$ and policy $\theta$ by
\[ \bm {J}_{\alpha}(s,\theta) =  \mathop{\lim\inf}_{\mathtt T\to\infty} \frac{1}{\mathtt T} \bm{E}_s^\theta\left[\sum_{k=1}^{\mathtt T} \alpha^k r(s_k,a_k)    \right], \]
and $\bm {J}_{\alpha}^*(s) \triangleq \sup_{\theta\in\Theta} \bm {J}_{\alpha}(s,\theta)$.
With the existence of stationary and deterministic optimal policy proved in Theorem~\ref{theorem:existence}, one may let
\[q(s) = \lim_{\alpha\to 1}\bm V_{\alpha}(s).  \]
with $\bm V_{\alpha}(s) = \bm {J}_{\alpha}^*(s) - \bm {J}_{\alpha}^*((0,0)).$

Then we show the monotonicity and submodularity of $q(s)$ by examining $\bm V_a(s)$. We do this by value iteration.
To this end, we define a dynamic programming operator $\bm{T}_{\alpha}$: given a measurable function $u:\mathbb{S}\mapsto \mathbb{R}$, let
\begin{align*}
\bm{T}_{\alpha}u(s) \triangleq \max_{a\in\mathbb{A}} \left[ r(s,a) + \alpha \bm G(u,s,a)  \right], \: s\in\mathbb{S}.
\end{align*}
Given $0< \alpha <1$, we define a function $\bm W'_\alpha:\mathbb{S}\mapsto[1,\infty)$ (depending on $\alpha$) that has exactly the same form as $\bm W(s)$ in~\eqref{eqn:FunctionW} but the parameters involved have a different constraint. Specifically, the equations~\eqref{eqn:lambdai}-\eqref{eqn:beta2} are replaced with
\begin{align*}
(1-\underline{\epsilon_i})(\lambda_i-1) <& \frac{1}{\alpha} - 1,  \\
1\leq \phi <& \frac{1}{\alpha},  \\
\phi \lambda_i^{N_i} [1 - (1-\underline{\epsilon_i})|A_i|^2 ] \geq& \phi + 1.
\end{align*}
Using the same arguments as for Lemma~\ref{lemma:boundedreward}, it is easy to see that $\big\|\sup_{a\in\mathbb{A}}r(s,a)\big\|_{\bm W'_{\alpha}} < \infty $. Thus, for any $0<\alpha<1$, $\|\bm {J}_{\alpha}^*(s)\|_{\bm W'_{\alpha}} < \infty$.
Furthermore, by some basic calculations, one obtains that $\bm W'_{\alpha}$ satisfies~\cite[Assumption 8.3.2]{hernandez1999further}. It then follows from Proposition 8.3.9 thereof, $\bm{T}_{\alpha}$ is contraction operator on $\mathbb{B}_{\bm W'_{\alpha}}(\mathbb{S})$. By Banach's Fixed Point Theorem,
for any $u\in\mathbb{B}_{\bm W'_{\alpha}}(\mathbb{S})$, $0<a<1$,
\begin{align} \label{eqn:BanachFixedPoint}
\lim_{n\to\infty}\bm{T}_{\alpha}^n u = \bm {J}_{\alpha}^*(s).
\end{align}
Since given $\alpha$, $\bm {J}_{\alpha}^*((0,0))$ is a constant, the structure (montonicity or submodularity) of $\bm V_{\alpha}(s)$ can be proved by showing that $\bm {J}_{\alpha}^*(s)$ has the same structure. By~\eqref{eqn:BanachFixedPoint}, it suffices to prove that the structure is preserved by the dynamic operator $\bm{T}_{\alpha}$.

\emph{Monotonicity:} Suppose $s\preccurlyeq s'$ and $u(s)\leq u(s')$, since for any $f$, $r(s,f(s)) \leq r(s',f(s'))$, it holds that
\[r(s,f(s)) + \alpha \bm G(u,s,f(s)) \leq r(s',f(s')) + \alpha \bm G(u,s',f(s'))\]
for any $f$, which yields $\bm{T}_{\alpha}u(s) \leq \bm{T}_{\alpha}u(s').$

\emph{Submodularity:}
By the monotonicity of $q(s)$, without any performance loss one may eliminate action $\mathbf{0}$. In the remainder, we let the action space $\mathbb{A}=\{e_1,e_2\}$.
Suppose $u\in\mathbb{B}_{\bm W'_{\alpha}}(\mathbb{S})$ is monotonic, and for any $s,s'\in\mathbb{S}$
\begin{align}  \label{eqn:submodular2}
u(s) + u(s') \geq u(s\downarrow s') + u(s \uparrow s'),
\end{align}
we need to prove $\bm{T}_{\alpha}u(s) + \bm{T}_{\alpha}u(s') \geq \bm{T}_{\alpha}u(s\downarrow s') + \bm{T}_{\alpha}u(s \uparrow s')$. By the definition of one stage reward function $r(s,a)$, it suffices to prove
\begin{align*}
&\max_{a\in\mathbb{A}}\bm G(u,s,a)  + \max_{a\in\mathbb{A}} \bm G(u,s',a)\\
\geq  & \max_{a\in\mathbb{A}} \bm G(u,s\downarrow s',a) + \max_{a\in\mathbb{A}} \bm G(u,s\uparrow s',a).
 \addtag \label{eqn:submodular3}
\end{align*}

Let $s=(j_1,j_2), s'=(j_1',j_2')$ with $j_1\leq j_1', j_2\geq j_2'$. Without loss of any generality, we assume $(1-\underline{\epsilon_1})(1-\epsilon_2) \geq (1-\epsilon_1) (1-\underline{\epsilon_2})$.
For the function $u$, define the optimal action associated with state $s$ by
\begin{align*}
a^*(s) \triangleq \mathop{\arg\max}_{{a\in\mathbb{A}}}\bm G(u,s,a).
\end{align*}
In the following, we prove~\eqref{eqn:submodular3} by cases.

\emph{Case $a^*(s\downarrow s') = a^*(s\uparrow s')$:} Without loss of generality, we let $a^*(s\downarrow s') = a^*(s\uparrow s')= e_1$. Let $\varepsilon_1 = (1-\underline{\epsilon_1})(1-\epsilon_2)$, one then obtains that
\begin{align*}
&\max_{a\in\mathbb{A}} \bm G(u,s,a)  + \max_{a\in\mathbb{A}} \bm G(u,s,a)\\
  \geq &\bm G(u,s,e_1)+\bm G(u,s',e_1) \\
 =  &\varepsilon_1 \Big(u((j_1+1,j_2+1)) + u((j_1'+1,j_2'+1))\Big) \\
 &+ (1-\underline{\epsilon_1})\epsilon_2 \Big(u((j_1+1,0)) + u((j_1'+1,0))\Big) \\
 & + \underline{\epsilon_1}(1-\epsilon_2) \Big(u((0,j_2+1)) + u((0,j_2'+1))\Big) \\
 &+ 2\underline{\epsilon_1}\epsilon_2 u((0,0)) \\
  \triangleq &\varepsilon_1 \Big(u((j_1+1,j_2+1)) + u((j_1'+1,j_2'+1))\Big) + \Lambda \\
  \geq& \varepsilon_1 \Big(u((j_1+1,j_2'+1)) + u((j_1'+1,j_2+1))\Big) + \Lambda \\
   = & \bm G(u,s\downarrow s',e_1) + \bm G(u,s\uparrow s',e_1) \\
 = & \max_{a\in\mathbb{A}} \bm G(u,s\downarrow s',a) + \max_{a\in\mathbb{A}}
 \bm G(u,s\uparrow s',a),
\end{align*}
where the second inequality follows from~\eqref{eqn:submodular2}.

\emph{Case $a^*(s\downarrow s') = e_1, a^*(s\uparrow s') = e_2$:} Let $\varepsilon_2 = (1-\epsilon_1) (1-\underline{\epsilon_2})$ and $\varepsilon_3=(1-\underline{\epsilon_1})\epsilon_2 - (1-\epsilon_1)\underline{\epsilon_2}$, one then obtains that
\begin{align*}
&\bm G(u,s,e_2)+\bm G(u,s',e_1)-
\bm G(u,s\hspace{-1mm}\downarrow\hspace{-1mm} s',e_1)-\bm G(u,s\hspace{-1mm}\uparrow\hspace{-1mm} s',e_2    )\\
= &\varepsilon_2 u((j_1+1,j_2+1)) + \varepsilon_1 u((j_1'+1,j_2'+1)) \\
&- \varepsilon_1 u((j_1+1,j_2'+1)) -\varepsilon_2 u((j_1'+1,j_2+1)) \\
&+ \varepsilon_3\Big(u((j_1'+1,0))-u((j_1+1,0))\Big)  \\
\geq &\varepsilon_1 u((j_1+1,j_2+1)) + \varepsilon_1 u((j_1'+1,j_2'+1)) \\
&- \varepsilon_1 u((j_1+1,j_2'+1)) -\varepsilon_1 u((j_1'+1,j_2+1)) \\
\geq &0,
\end{align*}
where the first inequality follows from the monotonicity of $u$ and the fact $\varepsilon_1 \geq \varepsilon_2$, and the second inequality is due to~\eqref{eqn:submodular2}. Equation~\eqref{eqn:submodular3} thus follows.

\emph{Case $a^*(s\downarrow s') = e_2, a^*(s\uparrow s') = e_1$:} One has the following:
\begin{align*}
&\bm G(u,s,e_2)+\bm G(u,s',e_1)-
\bm G(u,s\hspace{-1mm}\downarrow\hspace{-1mm} s',e_2)-\bm G(u,s\hspace{-1mm}\uparrow\hspace{-1mm} s',e_1    )\\
= &\varepsilon_2 u((j_1+1,j_2+1)) + \varepsilon_1 u((j_1'+1,j_2'+1)) \\
&- \varepsilon_2 u((j_1+1,j_2'+1)) -\varepsilon_1 u((j_1'+1,j_2+1)) \\
& + (\varepsilon_1-\varepsilon_2 + \epsilon_2-\underline{\epsilon_2})
\Big(u((0,j_2+1))-u((0,j_2'+1))\Big)  \\
\geq& \varepsilon_1 u((j_1+1,j_2+1)) + \varepsilon_1 u((j_1'+1,j_2'+1)) \\
&- \varepsilon_1 u((j_1+1,j_2'+1)) -\varepsilon_1 u((j_1'+1,j_2+1)) \\
& + (\varepsilon_1-\varepsilon_2) \Big( u((0,j_2+1)) +  u((j_1+1,j_2'+1)) \\
& - u((0,j_2'+1)) - u((j_1+1,j_2+1)) \Big)\\
\geq &0,
\end{align*}
which yields~\eqref{eqn:submodular3}. The proof thus is complete. \hfill $\square$
\end{proof}

We are ready to prove Theorem~\ref{theorem:OptimalAction}.
First, let fix $j_2$ and show that if $  f^*(s=(j_1,j_2)) = e_1$, then $f^*(s=(j_1+j,j_2)) = e_1$ with $j\in\mathbb{N}$.
Since $  f^*(s=(j_1,j_2)) = e_1$ implies that
\begin{align*}
&(\varepsilon_1-\varepsilon_2) q((j_1+1,j_2+1)) + \varepsilon_3q((j_1+1,0)) \\
\geq &\varepsilon_4 q((0,j_2+1))  + (\epsilon_1\underline{\epsilon_2} - \underline{\epsilon_1}\epsilon_2) q((0,0)) \\
\triangleq &\Lambda_3.
\end{align*}
where $\varepsilon_4=\epsilon_1(1-\underline{\epsilon_2}) - \underline{\epsilon_1}(1-\epsilon_2)$
Since $\varepsilon_1-\varepsilon_2 \geq 0$, $\varepsilon_3$ and $\Lambda_3$ is constant for a given $j_2$, by the monotonicity of $q$ in Lemma~\ref{lemma:differentialEqn}, one obtains that
$$(\varepsilon_1-\varepsilon_2) q((j_1+j+1,j_2)) + \varepsilon_3 q((j_1+j+1,0))
\geq\Lambda_3,
$$
which yields $f^*(s=(j_1+j,j_2)) = e_1$. Then it concludes that given a $j_2$, there is a critical curve $l_1(j_2)$ such that
 \begin{align*}
  f^*(s=(j_1,j_2))
  =&\left\{
        \begin{array}{ll}
           e_1, & \text{if $j_1 \geq l_1(j_2)$}, \\
           e_2, & \text{if $j_1 < l_1(j_2)$}.
        \end{array}
    \right.  \addtag  \label{eqn:ActionStructure2}
\end{align*}

Similarly, let fix $j_1$ and show that if $  f^*(s=(j_1,j_2)) = e_2$, then $f^*(s=(j_1,j_2+j)) = e_2$ with $j\in\mathbb{N}$. Note that $f^*(s=(j_1,j_2)) = e_2$ implies that
\begin{align*}
&\varepsilon_4 q((0,j_2+1))-(\varepsilon_1-\varepsilon_2) q((j_1+1,j_2+1))   \\
\geq& \varepsilon_3 q((j_1+1,0))
 + (\epsilon_1\underline{\epsilon_2} - \underline{\epsilon_1}\epsilon_2) q((0,0)) \\
\triangleq& \Lambda_4.   
\end{align*}
Then one has
\begin{align*}
&\varepsilon_4 q((0,j_2+j+1))-(\varepsilon_1-\varepsilon_2) q((j_1+1,j_2+j+1))\\
=& (\varepsilon_1-\varepsilon_2) \Big(q((0,j_2+j+1)) - q((j_1+1,j_2+j+1))\Big)  \\
&+ (\epsilon_2 - \underline{\epsilon_2}) q((0,j_2+j+1)) \\
\geq & (\varepsilon_1-\varepsilon_2) \Big(q((0,j_2+1)) - q((j_1+1,j_2+1))\Big)  \\
&+ (\epsilon_2 - \underline{\epsilon_2}) q((0,j_2+1)) \\
=&\varepsilon_4 q((0,j_2+1))-(\varepsilon_1-\varepsilon_2) q((j_1+1,j_2+1))\\
 \geq& \Lambda_4,
\end{align*}
where the first inequality follows from the monotonicity and submodularity of $q(s)$ established in Lemma~\ref{lemma:differentialEqn}.
Hence $f^*(s=(j_1,j_2+j)) = e_2$. Similarly, it concludes that given a $j_1$, there is a critical curve $l_2(j_1)$ such that
\begin{align*}
  f^*(s=(j_1,j_2))
  =&\left\{
        \begin{array}{ll}
           e_2, & \text{if $j_2 \geq l_2(j_1)$}, \\
           e_1, & \text{if $j_2 < l_2(j_1)$}.
        \end{array}
    \right.   \addtag \label{eqn:ActionStructure1}
\end{align*}
To simultaneously satisfy both~\eqref{eqn:ActionStructure2}~and~\eqref{eqn:ActionStructure1}, both functions $l_1(\cdot)$ and $l_2(\cdot)$ must be monotonically non-decreasing. Then the statements in Theorem~\ref{theorem:OptimalAction} follow immediately by letting $l_c(j_1,j_2)=l_2(j_1)-j_2$.

\section*{Appendix C\\ Proof of Theorem~\ref{theorem:indexpolicy}}
The byproduct of Theorem~\ref{theorem:OptimalAction} is that for an optimal policy at no time the action $\mathbf{0}$ is chosen. This can be extended to a general case, i.e., the constraint that at each time the attacker can attack at most $\mathtt{N}$ of $\mathtt{M}$ systems is equivalent to the constraint that the attacker attacks \emph{exactly} $\mathtt{N}$ of $\mathtt{M}$ systems.
With this in mind,
we prove the theorem using the results in~\cite{whittle1988restless} on the restless multi-armed bandit problem. 

Recall that $d_i^*(j,z_i)$ is the optimal rule for the state $\tau_{k-1}^{(i)} = j$ with $j\in\mathbb{N}$ when the subsidy is $z_i$. We then have the following definition and lemma.

\begin{definition} \cite{whittle1988restless}
The $i$-th system is said to be indexable if for any $j\in\mathbb{N}$, $d_i^*(j,z_i) = 0$ implies $d_i^*(j,z_i') = 0$ with $z_i' \geq z_i$. The whole system is indexable if each system is indexable.
\end{definition}

\begin{lemma} \label{lemma:indexability}
The system introduced in Section~\ref{section:problem-setup} is indexable.
\end{lemma}
\begin{proof}
We show that each system is indexable.
For ease of notations, throughout this proof, we omit the subscript $i$. Denote $p(\cdot)$  as the resulted equilibrium probability distribution of the state when $\ell(z)=j^*$ (function $\ell(\cdot)$ is, recall, introduced in~\eqref{eqn:thresholdAsymptotic}\footnote{Notice that the subscript $i$ has been omitted.}).
Then due to the threshold structure in~\eqref{eqn:thresholdAsymptotic}, one obtains that
\begin{align}
\sum_{j=0}^{\infty} p(j) =& 1,    \label{eqn:equilibrium1}\\
p(j)
=&\left\{
        \begin{array}{ll}
           (1-\epsilon) p(j-1), & \text{if $1\leq j \leq j^*$}, \\
           (1-\underline{\epsilon}) p(j-1), & \text{if $j > j^*$}.
        \end{array}
    \right.  \label{eqn:equilibrium2}
\end{align}
Note that the averaged reward obtained by the attacker has two parts: the averaged subsidy $\bm R_s$ and the averaged estimation error $\bm R_e$:
\begin{align*}
\bm R_s =& \sum_{j=0}^{j^*-1} p(j) z, \\
\bm R_e = & \sum_{j=0}^{\infty} p(j) \mathrm{Tr}(h^j(\hat{P})).
\end{align*}

Now fix the subsidy $z$ and consider a suboptimal policy. The policy has a similar threshold structure as in~\eqref{eqn:thresholdAsymptotic} but with the switching threshold $0 \leq j^{\diamond}< j^*$. Denote the corresponding equilibrium probability distribution as $p'$, which is computed in a similar way as~\eqref{eqn:equilibrium1}\eqref{eqn:equilibrium2}. Then one has
\begin{align} \label{eqn:equiDistrelation}
 \sum_{j=0}^{j^*-1} p(j) > \sum_{j=0}^{j^{\diamond}-1} p'(j).
\end{align}
Denote the averaged subsidy and averaged estimation error as $\bm R_s'$ and $\bm R_e'$, respectively. Due to the optimality of $l(z)=j^*$, one obtains that $\bm R_s - \bm R_s' \geq \bm R_e' - \bm R_e$, i.e.,
\begin{align*}
\left[\sum_{j=0}^{j^*-1} p(j)  - \sum_{j=0}^{j^{\diamond}-1} p'(j)\right] z \geq \bm R_e' - \bm R_e.
\end{align*}
Then by~\eqref{eqn:equiDistrelation}, for any $z'\geq z$, it holds that
\begin{align*}
\left[\sum_{j=0}^{j^*-1} p(j)  - \sum_{j=0}^{j^{\diamond}-1} p'(j)\right] z' \geq \bm R_e' - \bm R_e,
\end{align*}
which means that for any subsidy $z'\geq z$, the optimal rule for the states $0\leq j< j^*$ is still ``not attack".
The proof thus is complete.
\hfill $\square$
\end{proof}

Asymptotic optimality of the index-based policy $\pi_{\rm d}$ stated in Theorem~\ref{theorem:indexpolicy} follows immediately
from the indexability established in Lemma~\ref{lemma:indexability} and~\cite[Conjecture]{whittle1988restless}\footnote{\cite{weber1990index} presented a counterexample to the conjecture, and provided an additional technical requirement to assure the asymptotic optimality.
Notice that, however, this technical requirement is inconsequential, as argued by the authors, since the cases where the indexability is not sufficient are ``extremely rare" and  the deviation from optimality (if exists) is ``minuscule". On the other hand, it is quite difficult to verify this technical requirement (which says that a differential equation describing the fluid approximation of the index-based policy has no limit cycles or chaotic behavior), we thus omit the verification here.}.

\bibliographystyle{plainnat}
\bibliography{xq_reference}

\end{document}